\documentclass[journal]{IEEEtran}
\usepackage{cite}
\usepackage{hyperref}

\ifCLASSINFOpdf
  \usepackage[pdftex]{graphicx}
  \graphicspath{{./img/}}
  \DeclareGraphicsExtensions{.pdf,.jpeg,.png}
\else
  \usepackage[dvips]{graphicx}
  \graphicspath{{./eps/}}
  \DeclareGraphicsExtensions{.eps}
\fi
\usepackage{amsmath}
\usepackage{amsfonts}
\usepackage{amsthm}
\usepackage{color}

\usepackage{array}

\ifCLASSOPTIONcompsoc
  \usepackage[caption=false,font=normalsize,labelfont=sf,textfont=sf]{subfig}
\else
  \usepackage[caption=false,font=footnotesize]{subfig}
\fi
\usepackage{url}

\usepackage{bbm} 

\theoremstyle{remark} \newtheorem{lemma}{Lemma}
\theoremstyle{remark} \newtheorem{definition}{Definition}
\theoremstyle{definition} \newtheorem{procedure}{Procedure}
\theoremstyle{plain} \newtheorem{theorem}{Theorem}
\theoremstyle{plain} 

\newcommand{\mm}{\mathrm{mm}}

\newcommand{\R}{\mathbb{R}}

\newcommand{\fate}{\mathbf{e}}
\newcommand{\fatu}{\mathbf{u}}
\newcommand{\fatv}{\mathbf{v}}
\newcommand{\fatw}{\mathbf{w}}
\newcommand{\fatx}{\mathbf{x}}
\newcommand{\faty}{\mathbf{y}}

\newcommand{\fatmu}{\boldsymbol{\mu}}
\newcommand{\E}{\mathrm{E}}

\newcommand{\Cov}{\mathrm{Cov}}
\newcommand{\Cor}{\mathrm{Cor}}
\newcommand{\Var}{\mathrm{Var}}
\newcommand{\vectorization}{\mathrm{vec}}
\newcommand{\FDR}{\mathrm{FDR}}
\newcommand{\TPR}{\mathrm{TPR}}

\newcommand{\tr}{\mathrm{tr}}

\DeclareMathOperator*{\argmax}{arg\,max}

\renewcommand{\theenumi}{\arabic{enumi}}

\renewcommand{\theenumii}{\alph{enumii}}

\begin{document}

\bstctlcite{IEEEexample:BSTcontrol}

%
\title{FDR-Corrected Sparse Canonical Correlation Analysis with Applications to Imaging Genomics}
%
%
%

\author{Alexej~Gossmann,~%
        Pascal~Zille,~%
        Vince~Calhoun,~%
        and~Yu-Ping~Wang%

        \thanks{A. Gossmann is with the Bioinnovation PhD Program, Tulane University, New Orleans, LA 70118 USA.}
\thanks{Y.-P. Wang and P. Zille are with the Department of Biomedical Engineering, Tulane University, New Orleans, LA 70118 USA.}
\thanks{V. Calhoun is with The Mind Research Network, University of New Mexico, Albuquerque, NM 87131 USA, and also with the Department of Electrical and Computer Engineering, The University of New Mexico, Albuquerque, NM 87131 USA.}}

\IEEEpubid{\begin{minipage}{\textwidth}\ \\[12pt] \centering
  Copyright \copyright 2018 IEEE\@. Accepted for publication in the IEEE Transactions on Medical Imaging. Personal use is permitted, but republication/redistribution requires IEEE permission.
  See \url{http://www.ieee.org/publications_standards/publications/rights/index.html} for more information.
\end{minipage}}


\maketitle

\begin{abstract}
  \noindent Reducing the number of false discoveries is presently one of the most pressing issues in the life sciences. It is of especially great importance for many applications in neuroimaging and genomics, where datasets are typically high-dimensional, which means that the number of explanatory variables exceeds the sample size. The false discovery rate (FDR) is a criterion that can be employed to address that issue. Thus it has gained great popularity as a tool for testing multiple hypotheses. Canonical correlation analysis (CCA) is a statistical technique that is used to make sense of the cross-correlation of two sets of measurements collected on the same set of samples (e.g., brain imaging and genomic data for the same mental illness patients), and sparse CCA extends the classical method to high-dimensional settings. Here we propose a way of applying the FDR concept to sparse CCA, and a method to control the FDR\@. The proposed FDR correction directly influences the sparsity of the solution, adapting it to the unknown true sparsity level. Theoretical derivation as well as simulation studies show that our procedure indeed keeps the FDR of the canonical vectors below a user-specified target level. We apply the proposed method to an imaging genomics dataset from the Philadelphia Neurodevelopmental Cohort. Our results link the brain connectivity profiles derived from brain activity during an emotion identification task, as measured by functional magnetic resonance imaging (fMRI), to the corresponding subjects' genomic data.
\end{abstract}


%
\IEEEpeerreviewmaketitle

\section{Introduction}
\label{sec:introduction}

\IEEEPARstart{C}{anonical} correlation analysis (due to Hotelling, \cite{Hotelling1936}), or CCA, is a classical statistical technique, which is used to make sense of the cross-correlation of two sets of measurements collected on the same set of samples. More precisely, given two sets of random variables, CCA identifies linear combinations of each, which have maximum correlation with each other. The coefficients of these linear combinations of features are called canonical vectors. Like many classical statistical techniques, CCA fails in high-dimensional settings, when the number of variables in either of the two cross-correlated datasets exceeds the number of samples.
For application to high-dimensional data, several methods of sparse canonical correlation analysis (sparse CCA, or sCCA) have been proposed, where sparsity is imposed on the canonical vectors (e.g., see \cite{Witten2009, Witten2009extensions, Parkhomenko2009}). Many applications have demonstrated the usefulness of sparse CCA methods.
For example they are commonly used in genomics to analyze datasets consisting of two genomic assays for the same set of subjects or cells (e.g., \cite{Le_Cao2009-nk, Waaijenborg2008-qy} among many others). They have also been successfully employed for the analysis of neuroimaging and imaging genomics datasets (e.g., \cite{Lin2014-he, Du2016-iu}), such as brain imaging and DNA sequence data for the same set of brain disease or mental illness patients.

However, most of the widely used sparse CCA methods determine the level of sparsity of the canonical vectors based on criteria of model fit. Different types of cross-validation procedures have been proposed for sparse CCA (e.g., \cite{Parkhomenko2009, Du2016-iu, Waaijenborg2008-qy}), sometimes incorporating criteria such as the AIC or BIC (e.g., \cite{Wilms2015-ra}), and a permutation based method is proposed in \cite{Witten2009extensions}. In some cases authors simply impose a certain level of sparsity on the solution vectors (e.g., \cite{Le_Cao2009-nk}), which heavily relies on the appropriateness of such prior assumptions. In general, the behavior of the selection procedures for the sparsity parameters in sparse CCA is not well understood, and there is a lack of theoretical guarantees regarding the recovery of an appropriate sparsity level of the sparse CCA solution.

In this work we propose a definition of false discovery rate (FDR) for canonical vectors, which is subsequently used as a statistical criterion to determine an appropriate sparsity level for a sparse CCA solution.
The proposed FDR criterion is a generalization of the conventional FDR \cite{Benjamini1995-qd} to canonical correlation analysis.
With the aim of reliably obtaining canonical vectors with FDR below a user-specified level $q$, we propose an \emph{FDR-corrected sparse CCA} procedure.
Up to a (small) proportion of false discoveries, which our method keeps on average far below a user-specified level $q$, the FDR-corrected sparse CCA method is shown to produce canonical vectors consisting only of features that are truly cross-correlated between the two analyzed datasets.
Roughly, our proposed procedure consists of the following two steps: (1) we use a subsample of the data to fit a (conventional) sparse CCA model; and (2) we use the remainder of the dataset to perform an FDR correction on the result of step (1).
\IEEEpubidadjcol

With a theoretical derivation as well as simulation studies we show that our procedure indeed keeps the FDR of the canonical vectors at or below a user-specified target level.
Additionally, we apply our method to an imaging genomics dataset from the Philadelphia Neurodevelopmental Cohort (PNC).
We use sparse CCA to identify relationships between the subjects' DNA sequence data, and their functional brain connectivity profiles, which are measures derived from neural activity during an emotion identification (EMID) task, whereby the neural activity is measured by functional magnetic resonance imaging (fMRI).
The general goal of the presented application example is to gain understanding on how individual subjects' genetic makeup relates to their EMID-related brain connectivity measures, and vice versa. Our real data results are consistent with those reported elsewhere in the literature, demonstrating the validity of the method.

\section{Canonical Correlation Analysis}

In this section we describe the data structure and the notation that we use throughout the rest of the paper.
This section also introduces the classical CCA method, as well as the standard approach to sparse CCA.

\subsection{Assumed data structure and notation}
\label{sec:data_structure}

Let $\fatx^{(1)}, \fatx^{(2)}, \dots, \fatx^{(n)} \in\R^{p_X}$ be independent $\mathcal{N}(0, \Sigma_X)$-distributed vectors, and let $\faty^{(1)}, \faty^{(2)}, \dots, \faty^{(n)} \in \R^{p_Y}$ be independent $\mathcal{N}(0, \Sigma_Y)$-distributed vectors, where $\Sigma_X\in\R^{p_X\times p_X}$ and  $\Sigma_Y\in\R^{p_Y\times p_Y}$ are symmetric positive definite.

Assume that for all $k\in\left\{ 1,\dots,n \right\}$, the cross-covariance matrix
$$\Cov\left(\fatx^{(k)}, \faty^{(k)}\right) = \Sigma_{XY} \in \R^{p_X\times p_Y}$$
has entries
$\rho^{XY}_{i,j}$ for $i\in\left\{ 1, \dots, p_X \right\}$ and $j\in\left\{ 1,\dots,p_Y \right\}$,
and assume that
$$\Cov\left(\fatx^{(k)}, \faty^{(l)}\right) = 0, \mathrm{\,if\,} k \neq l.$$
Similarly, for $i,j\in\left\{ 1,\dots, p_X \right\}$ we denote the entries of $\Sigma_X$ by $\rho^X_{i,j}$, and the entries of $\Sigma_Y$ by $\rho^Y_{i,j}$ for $i,j\in\left\{ 1,\dots,p_Y \right\}$.

We define random matrices $X\in\R^{n\times p_X}$ and $Y\in\R^{n\times p_Y}$ by
\begin{equation*}
  X :=
  \begin{bmatrix}
    \left(\fatx^{(1)}\right)^T \\
    \left(\fatx^{(2)}\right)^T \\
    \vdots \\
    \left(\fatx^{(n)}\right)^T
  \end{bmatrix},
  \quad
  Y :=
  \begin{bmatrix}
    \left(\faty^{(1)}\right)^T \\
    \left(\faty^{(2)}\right)^T \\
    \vdots \\
    \left(\faty^{(n)}\right)^T
  \end{bmatrix}.
\end{equation*}

Thus, we can think of $X$ and $Y$ as two datasets containing respectively $p_X$ and $p_Y$ features collected for $n$ independent samples, where the features in the two datasets are cross-correlated with cross-covariance matrix $\Sigma_{XY}$. In that sense, the matrix $\begin{bmatrix} X & Y \end{bmatrix}$ is the combined dataset with covariance structure given by the matrix

\begin{equation}
  \Sigma :=
  \begin{bmatrix}
    \Sigma_X & \Sigma_{XY} \\
    \Sigma_{XY}^T & \Sigma_Y
  \end{bmatrix}.
  \label{eq:definition_of_Sigma}
\end{equation}

Please note that, even though all presented analytical derivations pre-suppose Gaussian data, we investigate departures from this distributional assumption in the simulation studies of Section~\ref{sec:hybrid_simulation}.

\subsection{Classical CCA}
\label{sec:classical_CCA}

Standard formulation of CCA \cite{Hotelling1936} seeks for vectors $\fatu\in\R^{p_X}$ and $\fatv\in\R^{p_Y}$ to maximize the sample correlation between $X\fatu$ and $Y\fatv$. Thus, the CCA optimization problem is given by
\begin{equation*}
  \argmax_{\fatu\in\R^{p_X}, \fatv\in\R^{p_Y}} \widehat{\Cov}(X\fatu, Y\fatv) = \argmax_{\fatu\in\R^{p_X}, \fatv\in\R^{p_Y}} \frac{1}{n} \fatu^T X^T Y \fatv,
\end{equation*}
subject to
\begin{align*}
  \widehat{\Var}(X\fatu) &= \frac{1}{n-1} \fatu^T X^T X \fatu = 1,\\
  \widehat{\Var}(Y\fatv) &= 1.
\end{align*}
The solution to this optimization problem, $\left(\widehat{\fatu}, \widehat{\fatv}\right)$, is called the first pair of canonical vectors. The linear combinations of features, $X\widehat{\fatu}$ and $Y\widehat{\fatv}$, are called the first pair of canonical variates. Subsequent pairs of canonical variates are restricted to be uncorrelated with the previous ones. For more detail we refer to \cite{Johnson2002-jx}.

\subsection{Sparse CCA}

The conventional CCA becomes degenerate if $n \leq \max\left\{ p_X, p_Y \right\}$, which is often the case in neuroimaging and other biomedical applications.
Sparse CCA (e.g., \cite{Parkhomenko2009, Witten2009, Witten2009extensions}) extends CCA to high-dimensional data by imposing a sparsity assumption on the CCA solution. Sparsity is achieved by utilizing penalty terms, such as the $\ell_1$-norm, on the canonical vectors. This results in a unique solution even when $p_X, p_Y \gg n$.

Most commonly the $\ell_1$-penalty is used in order to enforce the sparsity assumption on the canonical vectors. The resulting penalized CCA problem, which is introduced in \cite{Witten2009}, is given by
\begin{equation}
  \begin{gathered}
    \argmax_{\fatu\in\R^{p_X}, \fatv\in\R^{p_Y}} \frac{1}{n} \fatu^T X^T Y \fatv,\\
    \text{subject to}\\
    \| \fatu \|_2^2 \leq 1, \| \fatv \|_2^2 \leq 1, \| \fatu \|_1 \leq c_1, \| \fatv \|_1 \leq c_2.
    \label{eq:ordinary_SCCA}
  \end{gathered}
\end{equation}
However, the determination of the sparsity level, or the model tuning parameters $c_1$ and $c_2$, remains a challenging problem, and a topic of ongoing research, as briefly discussed in Section~\ref{sec:introduction}.

Higher-order pairs of canonical vectors can be found by applying sparse CCA to a residual matrix, obtained from $X^T Y$ and the previously found canonical variates (see \cite{Witten2009}).

\section{Defining false discovery rate (FDR) for sparse CCA}

It is not clear how the well-known definition of FDR (due to \cite{Benjamini1995-qd}), which is widely used in multiple hypothesis testing, can be carried over to sparse CCA\@.
In this work, we first propose an adaptation of the FDR concept to the context of sparse CCA, and then introduce a method which ensures that the FDR is kept below a given threshold.
Since we have a pair of canonical vectors $\fatu$ and $\fatv$, we consider the FDR in $\fatu$ and in $\fatv$ separately.
In most applications of sparse CCA the matrices $X$ and $Y$ will typically correspond to two totally different types of features (such as genomics vs.\ brain imaging), and mixing them within one common set of variables with a combined FDR does not seem appropriate.

For the sake of clarity, we present in the following the FDR derivation for the canonical vector $\fatu$ only. An identical derivation can be used for the canonical vector $\fatv$.

The population-level formulation of CCA \cite{Hotelling1936} seeks to maximize $\Cov(\fatx^T\fatu, \faty^T\fatv)$ under the constraints that $\Var(\fatx^T\fatu) = 1 = \Var(\faty^T\fatv)$, where $\fatx \sim \mathcal{N}(0, \Sigma_X)$ and $\faty \sim \mathcal{N}(0, \Sigma_Y)$ are random vectors distributed as the rows of $X$ and $Y$ respectively (cf., the sample formulation of CCA in Section~\ref{sec:classical_CCA}).
Let $\widehat{\fatu}$ be an estimate of the canonical vector $\fatu$.
For each ${i \in\left\{ 1,2,\dots,p_X \right\}}$, we say that the estimated coefficient $\hat{u}_i$ represents a \emph{false discovery} of the $i$th feature of $X$, if $\hat{u}_i \neq 0$ but $u_i$ does not affect the value of $\Cov(\fatx^T\fatu, \faty^T\fatv)$.
Notice that the coefficient $u_i$ contributes to $\Cov(\fatx^T\fatu, \faty^T\fatv)$ only through the multiplicative term
\begin{equation*}
  u_i \cdot \Cov(x_i, \faty^T \fatv) = u_i \cdot \sum_{j=1}^{p_Y} v_j \rho_{i,j}^{XY}.
\end{equation*}
It follows that the estimated coefficient $\hat{u}_i$ represents a false discovery of the $i$th feature of $X$, if and only if $\hat{u}_i \neq 0$ but $\sum_{j=1}^{p_Y} v_j \rho_{i,j}^{XY} = 0$ (or equivalently, if $\hat{u}_i \neq 0$ but the $i$th element of $\fatx$ is uncorrelated with the canonical variable $\faty^T \fatv$).

Thus, given the canonical vector $\fatv\in\R^{p_Y}$, the problem of identifying which non-zero entries of an estimate $\widehat{\fatu}$ represent false discoveries can be recast as testing the null hypothesis
\begin{equation}
  \mathrm{H}_i : \sum_{j=1}^{p_Y} v_j \rho_{i,j}^{XY} = 0,
  \label{eq:hypothesis_test_FDR_definition_section}
\end{equation}
for each $i \in\left\{ 1,2,\dots,p_X \right\}$ with $\hat{u}_i \neq 0$.

Let $R_{\widehat{\fatu}}$ be the number of non-zero entries in a sparse CCA estimate $\widehat{\fatu}$.
Let $V_{\widehat{\fatu}}$ denote the number of false rejections, that is, the number of indices $i \in\left\{ 1,2,\dots,p_X \right\}$, such that $\hat{u}_i \neq 0$ but $\sum_{j=1}^{p_Y} v_j \rho_{i,j}^{XY} = 0$.

Analogously, define $R_{\widehat{\fatv}}$ and $V_{\widehat{\fatv}}$ for the canonical vector $\widehat{\fatv}$.

\begin{definition}[False discovery rate]
  Define the false discovery rate in $\fatu$ as
  \begin{equation}
    \FDR(\widehat{\fatu}) := \E\left( \frac{V_{\widehat{\fatu}}}{\max\left\{ R_{\widehat{\fatu}}, 1 \right\}} \right),
    \label{eq:definition_FDRu}
  \end{equation}
  and analogously define,
  \begin{equation}
    \FDR(\widehat{\fatv}) := \E\left( \frac{V_{\widehat{\fatv}}}{\max\left\{ R_{\widehat{\fatv}}, 1 \right\}} \right).
    \label{eq:definition_FDRv}
  \end{equation}
  \label{def:FDR}
\end{definition}

In the next section we rely on this definition of FDR as an optimality criterion and as a guide in the identification of an appropriate sparsity level for the canonical vectors.

\section{FDR-corrected sparse CCA}
\label{sec:FDRcorrectedSCCA_intro}

Because the random vectors
$\left(
  \fatx^{(1)},
  \faty^{(1)}
\right)$,
$\left(
  \fatx^{(2)},
  \faty^{(2)}
\right)$, $\dots$,
$\left(
  \fatx^{(n)},
  \faty^{(n)}
\right)$, which form the rows of $X$ and $Y$, are independent and identically distributed,
we have that
\begin{equation*}
  \E(X^T Y \fatv) = n \cdot \Sigma_{XY} \cdot \fatv
\end{equation*}
for any $\fatv\in\R^{p_Y}$.
It follows that the null hypothesis $\mathrm{H}_i$ given by Equation (\ref{eq:hypothesis_test_FDR_definition_section}) is equivalent to
\begin{equation*}
  \mathrm{H}_i : \E\left(X^TY\fatv\right)_i = 0.
\end{equation*}

Motivated by that observation, the main idea of our FDR-correcting approach is to first obtain initial estimates $\widehat{\fatu}^{(0)}$ and $\widehat{\fatv}^{(0)}$ of the canonical vectors, and then to test null hypotheses of the form
\begin{equation}
  \mathrm{H}^{(u)}_i : \E \left(X^T Y \widehat{\fatv}^{(0)}\right)_i = 0,
  \label{eq:null_hypothesis_u}
\end{equation}
for each $i \in\left\{ 1,2,\dots,p_X \right\}$ with $\hat{u}_i \neq 0$, and
\begin{equation}
  \mathrm{H}^{(v)}_j : \E \left(Y^T X \widehat{\fatu}^{(0)}\right)_j = 0,
  \label{eq:null_hypothesis_v}
\end{equation}
for each $j \in\left\{ 1,2,\dots,p_Y \right\}$ with $\hat{v}_j \neq 0$, in order to determine which entries of $\fatu$ and $\fatv$ are truly non-zero.

\subsection{Asymptotic distribution}
\label{sec:asymptotic_distribution}

In order to be able to make probabilistic statements about the estimators of the sparse canonical vectors, we need to know the respective distributions of $X^T Y \fatv$ and $Y^T X \fatu$. In the following we will focus on $X^T Y \fatv$, but the obtained results clearly carry over to $Y^T X \fatu$, by swapping $X$ and $Y$ as well as $\fatu$ and $\fatv$.

\begin{theorem}[Asymptotic Normality]
  Let the random matrices $X$ and $Y$ be defined as above. For any vector $\fatv\in\R^{p_Y}$, it holds that
  \begin{equation}
    \sqrt{n} \left(\frac{1}{n} X^T Y \fatv - \fatmu \right) \overset{\mathcal{D}}{\longrightarrow} \mathcal{N}(0, \Omega),
    \label{eq:XtYv_is_asymptotically_normal}
  \end{equation}
  where $\fatmu\in\R^{p_X}$ has entries
  \begin{equation}
    \mu_i = \sum_{j = 1}^{p_Y} v_j \rho^{XY}_{i,j},
    \label{eq:XtYv_mean}
  \end{equation}
  and where $\Omega\in\R^{p_X \times p_X}$ has entries
  \begin{equation}
    \omega_{i,j} = \left( \sum_{k=1}^{p_Y} v_k \rho^{XY}_{i,k} \right) \left( \sum_{k=1}^{p_Y} v_k \rho^{XY}_{j,k} \right)
    + \rho^{X}_{i,j} \fatv^T \Sigma_Y \fatv.
    \label{eq:XtYv_covariance}
  \end{equation}
  \label{thm:XtYv_is_asymptotically_normal}
\end{theorem}

\begin{proof}
  The proof is given in Appendix~\ref{sec:XtYv_is_asymptotically_normal_proof}.
\end{proof}

\subsection{The FDR-corrected sparse CCA procedure}
\label{sec:FDRcorrectedCCA}

We propose an FDR correction procedure for sparse CCA based on the general idea outlined at the beginning of Section~\ref{sec:FDRcorrectedSCCA_intro}.
However several important aspects need to be taken into consideration when designing such a procedure.
In order to apply Theorem~\ref{thm:XtYv_is_asymptotically_normal} for purposes of testing hypotheses of the form (\ref{eq:null_hypothesis_u}) and (\ref{eq:null_hypothesis_v}), the preliminary estimates $\widehat{\fatu}^{(0)}$ and $\widehat{\fatv}^{(0)}$ as well as estimates of their associated variances (as given by Equation (\ref{eq:XtYv_covariance})) need to be obtained based on data that are independent from the dataset used for the hypotheses tests. Furthermore, the variances need to be estimated independently of $\widehat{\fatu}^{(0)}$ and $\widehat{\fatv}^{(0)}$, because otherwise the variance estimates corresponding to the non-zero entries of $\widehat{\fatu}^{(0)}$ and $\widehat{\fatv}^{(0)}$ will be inflated.

\begin{procedure}[FDR-corrected sparse CCA]
  \label{proc:FDRcorrectedSCCA}
  We propose an FDR-corrected sparse CCA procedure consisting of the following steps:

  \begin{enumerate}
    \item \emph{Data subsetting.} We divide each of the data matrices $X$ and $Y$ into three subsets of sizes $n_0$, $n_1$, and $n_2$, i.e.,
      \begin{equation*}
        X = \begin{bmatrix} X^{(0)} \\
                            X^{(1)} \\
                            X^{(2)}
            \end{bmatrix}
        \,\text{and}\,
        Y = \begin{bmatrix} Y^{(0)} \\
                            Y^{(1)} \\
                            Y^{(2)}
            \end{bmatrix},
        \,\text{where}\,
      \end{equation*}
      ${X^{(0)} \in \R^{n_0 \times p_X}}$, ${X^{(1)} \in \R^{n_1 \times p_X}}$, ${X^{(2)} \in \R^{n_2 \times p_X}}$, ${Y^{(0)} \in \R^{n_0 \times p_Y}}$, ${Y^{(1)} \in \R^{n_1 \times p_Y}}$, ${Y^{(2)} \in \R^{n_2 \times p_Y}}$, $n_0 + n_1 + n_2 = n$.
      In practice the split needs to be random, and the three subsets have to follow the same distribution.

    \item \emph{Preliminary estimates of the canonical vectors obtained on the first subset.} We obtain the preliminary estimates $\widehat{\fatu}^{(0)}$and $\widehat{\fatv}^{(0)}$ by applying the $\ell_1$-penalized CCA, as given by Equation (\ref{eq:ordinary_SCCA}), with a liberal choice for the tuning parameters $c_1$ and $c_2$, to $X^{(0)}$ and $Y^{(0)}$. At this step we aim to capture all truly non-zero entries of the true canonical vectors $\fatu$ and $\fatv$ within the support of the preliminary estimates $\widehat{\fatu}^{(0)}$ and $\widehat{\fatv}^{(0)}$. However we require that neither $\widehat{\fatu}^{(0)}$ nor $\widehat{\fatv}^{(0)}$ has more than $n_2$ non-zero entries, which is mainly needed for Step 5 below. The exact strategy of choosing the penalty parameters in this step is presented for each considered application of the procedure, at the time when the respective application is discussed in Sections~\ref{sec:gaussian_data},~\ref{sec:hybrid_simulation}, and~\ref{sec:upenn_data}.

      \item \emph{Covariance structure estimation on the second subset.} We use the matrices $X^{(1)}$ and $Y^{(1)}$ to obtain $\widehat{\Sigma}^{(1)}$, the maximum likelihood estimate of the covariance matrix of $\begin{bmatrix} X & Y \end{bmatrix}$.

        \item \emph{Statement of the hypotheses to be tested.} Let $I_u^{(0)}$ and $I_v^{(0)}$ be the sets of indices corresponding to the non-zero entries of $\widehat{\fatu}^{(0)}$ and $\widehat{\fatv}^{(0)}$, i.e.,
          \begin{align*}
            I_u^{(0)} &:= \{i : \hat{u}^{(0)}_i \neq 0\},\\
            I_v^{(0)} &:= \{j : \hat{v}^{(0)}_j \neq 0\}.
          \end{align*}
          For all $i \in I_u^{(0)}$ and all $j \in I_v^{(0)}$ we intend to test whether $u_i$ and $v_j$ are truly non-zero.
          We slightly adjust the hypotheses tests of form (\ref{eq:null_hypothesis_u}) and (\ref{eq:null_hypothesis_v}) to the present framework, where for each $i \in I_u^{(0)}$ and each $j \in I_v^{(0)}$ we test,
          \begin{align}
            \mathrm{H}^{(u)}_i &:& \E \left[ \left( X^{(2)} \right)^T Y^{(2)} \widehat{\fatv}^{(0)} \middle| \Sigma = \widehat{\Sigma}^{(1)} \right]_i = 0,
            \label{eq:null_hypothesis_u_mod} \\
            \mathrm{H}^{(v)}_j &:& \E \left[ \left( Y^{(2)} \right)^T X^{(2)} \widehat{\fatu}^{(0)} \middle| \Sigma = \widehat{\Sigma}^{(1)} \right]_j = 0.
            \label{eq:null_hypothesis_v_mod}
          \end{align}
          Notice that the tests (\ref{eq:null_hypothesis_u_mod}) and (\ref{eq:null_hypothesis_v_mod}) are heuristics in the sense that they approximate the tests (\ref{eq:null_hypothesis_u}) and (\ref{eq:null_hypothesis_v}).

        \item \emph{P-value calculation using the third set.} Since $\hat{u}^{(0)}_i = 0$ and  $\hat{v}^{(0)}_j = 0$ for all $i \notin I_u^{(0)}$ and $j \notin I_v^{(0)}$, we can remove the $i$th column of $X^{(2)}$ and the $j$th column of $Y^{(2)}$ for every $i \notin I_u^{(0)}$ and every $j \notin I_v^{(0)}$, without affecting any of the products in Equations (\ref{eq:null_hypothesis_u_mod}) and (\ref{eq:null_hypothesis_v_mod}). Thus we can replace the matrices $X^{(2)}$ and $Y^{(2)}$ by their versions containing only the columns indexed by $I_u^{(0)}$ and $I_v^{(0)}$ respectively, which amounts to a dimensional reduction. Due to Step 2 the cardinality of each of these sets, $\lvert I_u^{(0)}\rvert$ and $\lvert I_v^{(0)} \rvert$, should be sufficiently smaller than $n_2$. Therefore we are now dealing with a low-dimensional problem, so that we are confident in applying Theorem~\ref{thm:XtYv_is_asymptotically_normal}.
          Now, denoting the $i$th entry of $\left( X^{(2)} \right)^T Y^{(2)} \widehat{\fatv}^{(0)}$ by
          $$\xi_i := \left[\left( X^{(2)} \right)^T Y^{(2)} \widehat{\fatv}^{(0)}\right]_i,$$
          \emph{under the null hypothesis} Theorem~\ref{thm:XtYv_is_asymptotically_normal} implies the approximation
          \begin{equation}
            \left(\frac{1}{\sqrt{n}} \xi_i \middle| \Sigma = \widehat{\Sigma}^{(1)} \right) \sim \mathcal{N}\left( 0, \widehat{\omega}_{i,i} \right),
            \label{eq:test_statistic}
          \end{equation}
          where $\widehat{\omega}_{i,i}$ is defined according to Equation (\ref{eq:XtYv_covariance}) with $\widehat{\mathbf{v}}^{(0)}$ and $\widehat{\Sigma}^{(1)}$ from Steps 2 and 3 substituted in place of $\mathbf{v}$ and $\Sigma$.
          Likewise, from Theorem~\ref{thm:XtYv_is_asymptotically_normal} we obtain an asymptotic distribution for $\left( \frac{1}{\sqrt{n}} \left( Y^{(2)} \right)^T X^{(2)} \widehat{\fatu}^{(0)} \middle| \Sigma = \widehat{\Sigma}^{(1)} \right)_i$ under the null hypothesis.
          We calculate p-values for the hypotheses tests (\ref{eq:null_hypothesis_u_mod}) and (\ref{eq:null_hypothesis_v_mod}) based on these asymptotic distributions.
        \item \emph{FDR correction.} After choosing a desired FDR level $q \in (0, 1)$, in order to control the FDR at the level $q$, we apply the Benjamini-Hochberg procedure to the p-values obtained in Step 5. The procedure is applied twice -- separately for the two sets of hypotheses $\left\{ H^{(u)}_i \middle| i \in I_u^{(0)}\right\}$ and $\left\{ H^{(v)}_j \middle| j \in I_v^{(0)}\right\}$.\footnote{Notice that the FDR correction can just as well be be applied with two different levels $q_u$ and $q_v$ for the two sets of hypotheses, or it can also be applied just once to the combined set of hypotheses $H^{(u)}_i$ and $H^{(v)}_j$ for $i \in I_u^{(0)}$ and $j \in I_v^{(0)}$, in order to control the FDR of the ``full model'' (controlling the FDR for each of $\fatu$ and $\fatv$ separately at level $q$ in general does not imply FDR control of the ``full'' model at the same level $q$).} This results in a new, FDR-corrected, solution $\widehat{\fatu}$ and $\widehat{\fatv}$, which is given up to a normalization constant by
      \begin{align*}
        \hat{u}_i &:= \begin{cases}
          \left(X^T Y \widehat{\fatv}^{(0)}\right)_i, &\quad\mathrm{for\,any\,rejected\,}H_i^{(u)},\\
          0, &\quad\mathrm{otherwise}.
        \end{cases}\\
        \hat{v}_j &:= \begin{cases}
          \left(Y^T X \widehat{\fatu}^{(0)}\right)_j, &\quad\mathrm{for\,any\,rejected\,}H_j^{(v)},\\
          0, &\quad\mathrm{otherwise}.
        \end{cases}
      \end{align*}
  \end{enumerate}
\end{procedure}

Of course every step in Procedure~\ref{proc:FDRcorrectedSCCA} would benefit from a larger size of the utilized subsample, and there is a trade-off between $n_0$, $n_1$, and $n_2$ to be made.
However, for simplicity of exposition, in what follows we always use subsamples of equal size with $n_0 = \lfloor{n/3}\rfloor$.

As evidenced by the derivation of the method, given that the true solution is sparse, the FDR correction step adapts the sparsity of the estimator to the unknown sparsity of the true solution.

\section{Simulation studies}

We performed a number of simulation studies, in order to evaluate the performance of Procedure~\ref{proc:FDRcorrectedSCCA}, the proposed FDR-corrected sparse CCA method, in both idealized and more realistic scenarios. The procedure was derived based on the assumption that $X$ and $Y$ have Gaussian entries. Here, we first present simulation results under such Gaussian scenarios, in order to verify that the proposed procedure indeed controls the FDR under the assumptions that its derivation relies on. Then we show simulation studies evaluating the performance of the method on non-Gaussian data, which are generated based on real single-nucleotide polymorphism (SNP) data, and closely resemble the data in the real imaging genomics application presented in Section~\ref{sec:upenn_data}.
All simulation studies, which do not involve human subject data, can be reproduced with the code available at \url{https://github.com/agisga/FDRcorrectedSCCA}.

\subsection{Simulation study with Gaussian data}
\label{sec:gaussian_data}

Although we consider both, low-dimensional ($n < p_X, p_Y$) and high-dimensional ($n > p_X, p_Y$) scenarios, in this section we present only the simulations with high-dimensional data due to the higher relevance of that case to the intended biomedical applications. Due to given space constraints the simulation studies for the low-dimensional case are found in the Supplementary Materials (unsurprisingly the performance of our method improves as the dimensionality of the problem decreases).
In all simulations of this section we use $n = 600$, and $p_X  = p_Y = 1500$.
We denote by $s_X$ the number of variables in $X$ that are truly cross-correlated (i.e., have been generated with a non-zero correlation) with some variables in $Y$. Likewise, we denote by $s_Y$ the number of variables in $Y$ that are truly cross-correlated with some variables in $X$. For the simulations in this section we consider all combinations between $s_X, s_Y \in \left\{1, 20, 40, 60, 80, 100, 120\right\}$.
The matrices $\Sigma^X$, $\Sigma^Y$, and $\Sigma^{XY}$ are block-wise constant, and are populated with the following entries.
\begin{equation}
  \label{eq:Sigma_constant}
  \begin{aligned}
    \rho^X_{i,j} &= \begin{cases}
      1, \quad&\mathrm{if}\quad i = j,\\
      0.5, \quad&\mathrm{if}\quad i \neq j, i \leq s_X, j\leq s_X,\\
      0.1, \quad&\mathrm{otherwise}.
    \end{cases} \\
    \rho^Y_{i,j} &= \begin{cases}
      1, \quad&\mathrm{if}\quad i = j,\\
      0.5, \quad&\mathrm{if}\quad i \neq j, i \leq s_Y, j\leq s_Y,\\
      0.1, \quad&\mathrm{otherwise}.
    \end{cases} \\
    \rho^{XY}_{i,j} &= \begin{cases}
      0.4, \quad&\mathrm{if}\quad i \leq s_X, j\leq s_Y,\\
      0, \quad&\mathrm{otherwise}.
    \end{cases}
  \end{aligned}
\end{equation}
In particular, the cross-correlation between features of $X$ and $Y$ is of only mediocre strength, having a magnitude of $0.4$.
Furthermore, the choice of parameters given in Equation (\ref{eq:Sigma_constant}) ensures that the covariance matrix
\begin{equation}
  \Sigma =
  \begin{bmatrix}
    \Sigma_X & \Sigma_{XY} \\
    \Sigma_{XY}^T & \Sigma_Y
  \end{bmatrix}
  \label{eq:Sigma}
\end{equation}
is positive definite.
A visualization of the block-wise constant structure of such a covariance matrix $\Sigma$ is given in Figure~\ref{fig:Sigma_constant}.

\begin{figure}[ht]
  \centering
  \includegraphics[width=0.30\textwidth]{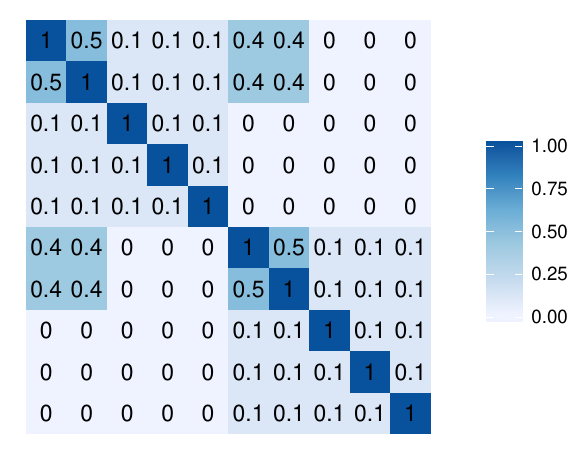}
  \caption{Shown is the block-wise constant structure of a covariance matrix $\Sigma$ as defined by equations (\ref{eq:Sigma_constant}) and (\ref{eq:Sigma}) with $p_X = p_Y = 5$ and $s_X = s_Y = 2$ (for clarity of representation the matrix dimensions are much smaller here than in the simulation study). The $x$- and $y$-axes represent the indices of the $(p_X + p_Y)$ columns in the matrix $\left[X\;\;Y\right]$
    and the magnitude of the covariance of each pair of variables is mapped to the color scale.}
  \label{fig:Sigma_constant}
\end{figure}

We perform Procedure~\ref{proc:FDRcorrectedSCCA} with three choices for the target FDR level $q = 0.05, 0.1, 0.2$. For Step 2 of the procedure (estimation of preliminary solution $\widehat{\fatu}^{(0)}, \widehat{\fatv}^{(0)}$) we use the $\ell_1$-penalized CCA of Equation (\ref{eq:ordinary_SCCA}), where we tune the penalty parameters such that $\widehat{\fatu}^{(0)}$ and $\widehat{\fatv}^{(0)}$ each contain approximately $\frac{n_2}{2} = \frac{n/3}{2} = 100$ non-zero entries. Enforcing the sparsity of $\widehat{\fatu}^{(0)}$ and $\widehat{\fatv}^{(0)}$ to be restricted in such a way results in a low-dimensional problem in Step 5 of the procedure (the p-value calculation) where we apply the asymptotic distribution of Theorem~\ref{thm:XtYv_is_asymptotically_normal}. However note that by following this approach we implicitly assume that $s_X, s_Y < \frac{n_2}{2} = 100$. This is true for most but not for all of the considered simulation settings, and we investigate the effects on the corresponding results below.

We compare the proposed method to the widely-used $\ell_1$-penalized CCA in conjunction with multiple strategies for the $\ell_1$-penalty parameter selection. An excellent implementation of the $\ell_1$-penalized CCA, as given in Equation~\ref{eq:ordinary_SCCA}, is available in the R package \texttt{PMA} \cite{Witten2009}, which uses the parametrization $c_1 = \lambda_1 \cdot \sqrt{p_X}$ and $c_2 = \lambda_2 \cdot \sqrt{p_Y}$ with $\lambda_1, \lambda_2 \in(0, 1)$ to be selected manually or by an automated procedure. When selecting an optimal penalty level in sparse CCA it is common practice to set $\lambda_1 = \lambda_2 =: \lambda$ (e.g., the default setting of the automated selection procedure \texttt{CCA.permute} in the R package \texttt{PMA} \cite{Witten2009}), because this substantially decreases the computational burden by reducing the 2D parameter search to a 1D search. For purposes of comparison to our method, we consider the following strategies for the selection of $\lambda$ in $\ell_1$-penalized CCA.
\begin{itemize}
  \item Fixed parameter values $\lambda = 0.01$ (very sparse) and ${\lambda = 0.3}$ (moderately sparse).
  \item Automated selection of an optimal $\lambda$ by the permutation based procedure of \cite{Witten2009}, where, as per the default settings of the R package \texttt{PMA} \cite{Witten2009}, $\lambda$ is selected from among 10 equispaced values between 0.1 and 0.7.
  \item A 5-fold cross-validation (CV) approach which selects the $\lambda$ value that maximizes the canonical correlation $\Cor(X\widehat{\fatu}, Y\widehat{\fatv})$ on the test set. As in the permutation based approach ten equispaced values between 0.1 and 0.7 are considered for $\lambda$.
\end{itemize}
Notice that these methods are performed on the whole data, while Procedure~\ref{proc:FDRcorrectedSCCA} splits the data into three subsets, effectively performing the CCA step on only a third of the data.
Nonetheless Procedure~\ref{proc:FDRcorrectedSCCA} outperforms these competing methods in many respects as we will see below.

The simulation is performed 500 times for every combination of the considered methods and parameter values. In every setting we estimate $\FDR(\widehat{\fatu})$ and $\FDR(\widehat{\fatv})$ as the mean values of the false discovery proportions (FDP) obtained from the respective 500 simulation runs, whereby $\mathrm{FDP}(\widehat{\fatv}) := \left( V_{\widehat{\fatv}} / \max\left\{ R_{\widehat{\fatv}}, 1 \right\} \right)$ (and likewise $\mathrm{FDP}(\widehat{\fatu})$, see Equations (\ref{eq:definition_FDRu}) and (\ref{eq:definition_FDRv})).
We also record some other statistics from the simulations, such as the estimated true positive rate (TPR, also known as sensitivity), which is the expected proportion of correctly identified features among all the truly significant features.
Like FDR, TPR is computed separately for $\widehat{\fatu}$ and $\widehat{\fatv}$, and is denoted by $\TPR\left(\widehat{\fatu}\right)$ and $\TPR\left(\widehat{\fatv}\right)$ respectively. The true value of $\TPR\left(\widehat{\fatv}\right)$ is estimated as an average of the true positive proportions (TPP) observed in the 500 respective simulation runs, whereby $\mathrm{TPP}(\widehat{\fatv}) := \left(R_{\widehat{\fatv}} - V_{\widehat{\fatv}}\right) / s_Y$ (and analogously $\mathrm{TPP}(\widehat{\fatu})$).
Naturally, it is desirable to have a high TPR and a low FDR, but there is a trade-off between the two quantities.

Figure~\ref{fig:one-cross-cor-block-Sigma-sim-FDR} shows the estimated $\FDR(\widehat{\fatv})$, and Figure~\ref{fig:one-cross-cor-block-Sigma-sim-TPR} shows the estimated $\TPR(\widehat{\fatv})$ for every considered method at every considered simulation setting. Because $X$ and $Y$ are generated by the same procedure with the same dimensions and with the same values considered for $s_X$ and $s_Y$, and because the proposed method is symmetric in the estimation of $\fatu$ and $\fatv$, the performance of the estimators $\widehat{\fatu}$ and $\widehat{\fatv}$ is exactly the same. Thus, for clarity of presentation we omit the inclusion of results for $\widehat{\fatu}$ in the figures of this section.

We observe that the estimated FDR of the proposed method consistently stays below the target upper bound $q$, being in fact a little too conservative by reaching at most a level of about $q/2$. This validates our method's FDR-correcting behavior, which is its intended purpose.

Since we have enforced that the preliminary estimates $\widehat{\fatu}^{(0)}$ and $\widehat{\fatv}^{(0)}$ each have about $\frac{n_2}{2} = 100$ non-zero entries, thus implicitly assuming that $s_X$ and $s_Y$ are sufficiently smaller than 100, the performance in simulations with $s_X, s_Y < 80$ differs from the performance observed when $s_X, s_Y \in \{80, 100, 120\}$. As expected we have that $\FDR(\widehat{\fatv}) \approx 0$ when $s_Y \geq 100$. Similarly we observe that as $s_X$ is increasing from 1 to 80, $\FDR(\widehat{\fatv})$ grows towards the nominal level $q$ and $\TPR(\widehat{\fatv})$ grows towards 1. The growth behavior is however significantly slowed down or stagnates when $s_X \geq 80$.

Compared to the FDR-corrected sparse CCA procedure, the other considered sparse CCA methods fail to consistently provide a satisfactory FDR level across all considered sparsity settings. Apart from our technique only the permutation based method of \cite{Witten2009} adapts to the true sparsity of the data, and yields an FDR that consistently stays at or below a level of about 0.2 regardless of the sparsity of the true solution. However, as evidenced by Figure~\ref{fig:one-cross-cor-block-Sigma-sim-TPR}, in most settings our procedure has a higher TPR than the permutation based method. Among all considered methods and across all sparsity settings the 5-fold CV based procedure results in the highest observed TPR values, but its FDR values are also among the largest ones. This can be expected because the utilized cross-validation criterion (which is similar to the one used in \cite{Parkhomenko2009}) aims to maximize model fit rather than to avoid false discoveries, i.e., the canonical correlation is maximized at the expense of the selection of many irrelevant features.
Moreover, it is evident from the two figures that $\ell_1$-penalized CCA with the penalty parameter $\lambda$ fixed at an arbitrary value (i.e., selected in a way that is not data-driven) will rarely produce the desired outcome in terms of maintaining a low FDR and a high TPR.

Additionally, we investigate the distribution of TPP values achieved by our method. When the cross-correlation structure between $X$ and $Y$ is as in the above example, where the cross-correlated variables fall into a single correlated block within $X$ and a single correlated block within $Y$ (as depicted in Figure~\ref{fig:Sigma_constant}), the sparse CCA procedure has a strong tendency to either select all of the cross-correlated features or none.
Likewise, when the cross-correlated features of $X$ (or $Y$) fall within $k$ distinct correlated blocks within $X$ (or $Y$), the TPP distribution seems to be discrete and concentrated at (at most) $(k+1)$ distinct values which are spaced apart at equal distances between 0 and 1.
That is, when the cross-correlated features fall into distinct correlation blocks within $X$ or $Y$, the method tends to either select an entire correlated block of variables within $X$ or $Y$, or to miss it entirely.
We confirmed this phenomenon with additional simulation studies, where the correlation structure within $X$ and within $Y$ is varied respectively.  Figure~\ref{fig:TPR_latent_variable_simulation} shows the TPP distribution, when the cross-correlated features, of which there are 60 in $X$ and 60 in $Y$, fall within $k \in \{1,2,3,15\}$ distinct correlation blocks within each matrix. The distributions' discrete character with at most $(k+1)$ peaks can be clearly observed, and as $k$ gets larger (e.g., $k = 15$) the TPP distribution appears more continuous. In fact a similar behaviour can also be observed for other sparse CCA methods.

Summing everything up, the simulation studies of this section support the theoretical derivation of the FDR-corrected sparse CCA method of Procedure~\ref{proc:FDRcorrectedSCCA}. As expected, the method controls the false discovery rates at the specified level $q$ very well, when applied to Gaussian data. This adds empirical evidence to the properties implied by the mathematical derivation.
Moreover we observe that Procedure~\ref{proc:FDRcorrectedSCCA} achieves the best trade-off between FDR and TPR when compared to several other commonly-used sparse CCA methods. It generally outperforms the permutation based method of \cite{Witten2009} and the very sparse $\ell_1$-penalized CCA ($\lambda = 0.01$) in terms of both FDR and TPR; and while some other methods (the $\ell_1$-penalized CCA with fixed $\lambda = 0.3$, or in conjunction with  5-fold CV) generally achieve a higher TPR, they perform poorly in terms of FDR in comparison to the proposed method.

\begin{figure*}[ht]
\centering
\subfloat[Estimated FDR by true sparsity of $\fatu$ and $\fatv$.]{\includegraphics[width=0.9\textwidth]{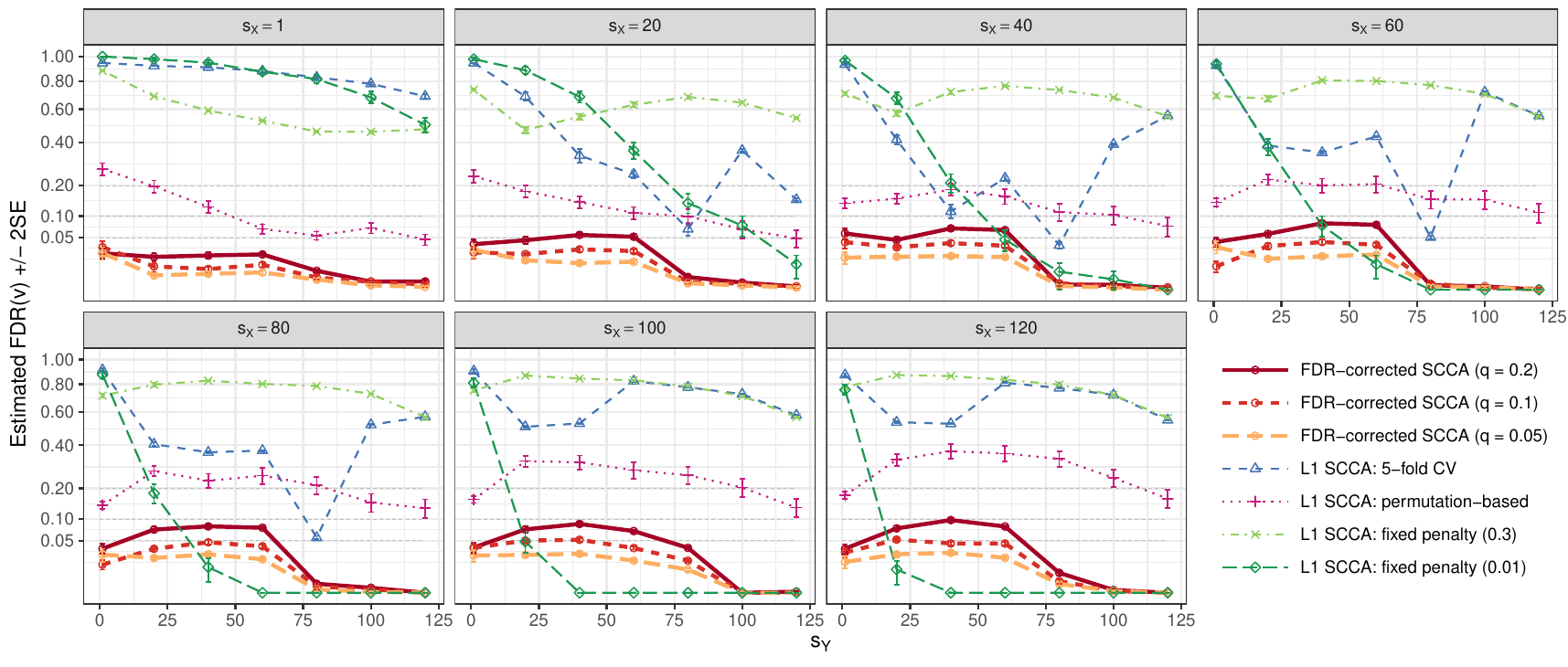}%
\label{fig:one-cross-cor-block-Sigma-sim-FDR}}
\hfil
\subfloat[Estimated TPR by true sparsity of $\fatu$ and $\fatv$.]{\includegraphics[width=0.9\textwidth]{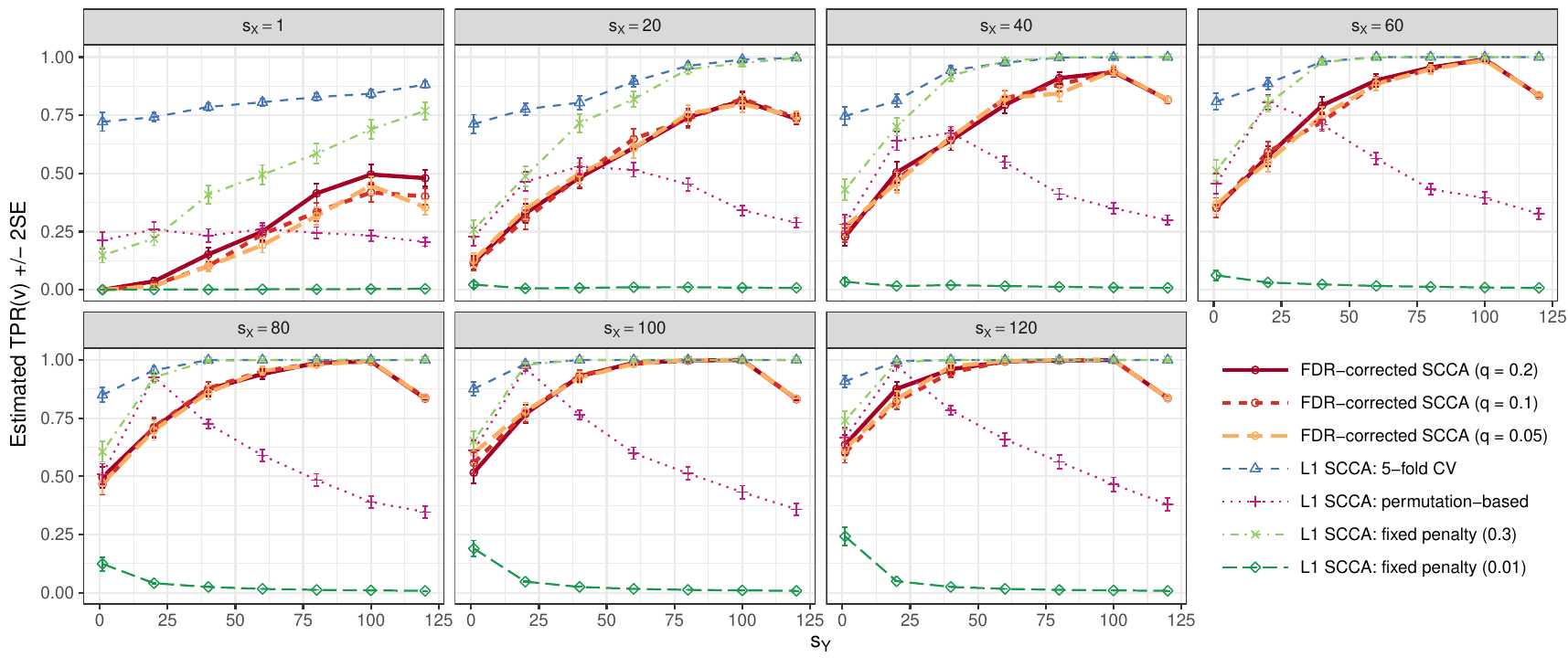}%
\label{fig:one-cross-cor-block-Sigma-sim-TPR}}
\caption{The performance of Procedure~\ref{proc:FDRcorrectedSCCA}, the proposed FDR-corrected sparse CCA procedure, at different target FDR levels $q = 0.05, 0.1, 0.2$ is compared to the $\ell_1$-penalized sparse CCA (L1 SCCA) with the penalty parameters either set to a fixed value (high sparsity: $\lambda = 0.01$, or moderate sparsity: $\lambda = 0.3$), or selected via the permutation based method of \cite{Witten2009}, or selected via a 5-fold cross-validation procedure (an approach based on \cite{Parkhomenko2009}). The matrices $X, Y \in \R^{600\times 1500}$ have Gaussian entries. All values are averaged over 500 simulation runs performed at each combination of parameters, and  all error bars correspond to $\pm2\mathrm{SE}$.
  (a) The estimated false discovery rate (FDR) in $\widehat{\fatv}$ is shown (on a square root scale) with respect to changes in $s_X$, the number of nonzero entries of $\fatu$, and $s_Y$, the number of non-zero entries of $\fatv$. The estimated FDR of Procedure~\ref{proc:FDRcorrectedSCCA} always stays below the specified upper bound $q$, with a sharp decrease at about $s_Y = 70$ due to the fact that the sparsity of the preliminary solution $\widehat{\fatv}^{(0)}$ was constrained to be approximately equal to 100 in the application of Procedure~\ref{proc:FDRcorrectedSCCA}. Estimated FDR of other methods is either generally inflated, or not adaptive to changes in $s_X$ and $s_Y$.
  (b) The estimated true positive rate (TPR, or sensitivity) of $\widehat{\fatv}$ is shown with respect to changes in $s_X$ and $s_Y$. Procedure~\ref{proc:FDRcorrectedSCCA} is competitive with the other methods in terms of TPR, even though its FDR is much lower than that of other methods (it achieves the best FDR-TPR trade-off). The TPR curves of Procedure~\ref{proc:FDRcorrectedSCCA} have a dip at $s_Y = 100$ because the sparsity of the preliminary solution $\widehat{\fatv}^{(0)}$ was constrained to be approximately equal to 100 in the procedure's application.}
\label{fig:one-cross-cor-block-Sigma-sim}
\end{figure*}

\begin{figure}[ht]
  \centering
  \includegraphics[width=0.45\textwidth]{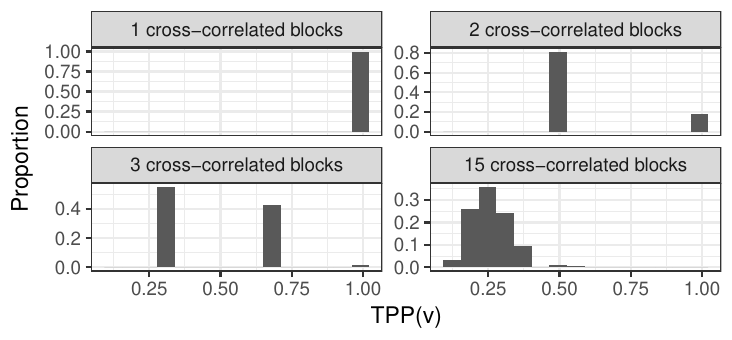}
  \caption{Shown is the empirical distribution of $\mathrm{TPP}(\widehat{\fatv})$ across 500 applications of the FDR-corrected sparse CCA method (Procedure~\ref{proc:FDRcorrectedSCCA}). The matrices $X, Y \in \R^{900\times 900}$ have Gaussian entries. The cross-correlated features, of which there are 60 in $X$ and 60 in $Y$, fall within $k \in\{1, 2, 3, 15\}$ distinct correlation blocks within each matrix. The resulting TPP distributions are discrete and concentrated at $(k+1)$ or fewer values, while appearing more continuous for larger $k$.}
  \label{fig:TPR_latent_variable_simulation}
\end{figure}

\subsection{Simulation study with non-Gaussian data (investigating robustness to distributional assumptions)}
\label{sec:hybrid_simulation}

In this section we consider more realistic simulation studies in the context of imaging genomics.
The matrix $X$ represents SNP data, which are not normally distributed, and the matrix $Y$ represents approximately Gaussian features obtained from brain fMRI data.
There are $n = 956$ subjects in the dataset (i.e., the number rows in $X$ and in $Y$).
We use real SNP data, which are also used in the real data analysis of Section~\ref{sec:upenn_data} (see Section~\ref{sec:data_acquisition} for more detail about the data), to generate $X$, as described in the following.
First, within each gene, PCA is performed on the SNPs that correspond to that gene; i.e., for each gene, PCA is performed separately on a sub-matrix consisting of only those columns of the full SNP matrix which contain SNP data corresponding to that particular gene.
Then, within each obtained gene-specific sub-matrix of principle components, only those principal components that explain at least 75\% of the variance within that gene are kept as features for further analysis.
The sub-matrices of features that were retained for each gene, when concatenated, constitute the full matrix of genomic features.
Many of the resulting features are clearly non-Gaussian, because many variables are heavily skewed, and many are clearly discrete.\footnote{For example, of the resulting 60372 variables, 8356 attain 10 or fewer distinct values, and (after standardization to sample standard deviation of 1) the method of moments estimator of skewness is greater than 1 for about 10\% of all variables, and greater than $\frac{1}{2}$ for about 28\% of all variables.}
Of the generated $60372$ genomic features, $p_X = 10000$ variables are randomly selected to form the matrix $X$ in each repetition of the simulation. The matrix $Y$ is randomly sampled with independent standard normal entries, consisting of $p_Y = 10000$ columns.
We cross-correlate the data in $X$ and $Y$ via the following procedure. Given, the number of cross-correlated features, ${s\in\{1,2,\ldots,p\}}$, where $p := p_X = p_Y$, and a parameter $\rho^{XY}\in(0, 1)$, we first generate a latent variable $\mathbf{z}\in\R^n$ with standard normal entries. Then, for each $i \in\{1,2,\ldots,s\}$, we replace the column $\mathbf{x}_i$ of $X$ with
\begin{equation*}
  \left(\sqrt{1-\rho^{XY}}\right)\mathbf{x}_i + \left(\sqrt{\rho^{XY}}\right) \mathbf{z},
\end{equation*}
and likewise we replace $\mathbf{y}_i$ with
\begin{equation*}
  \left(\sqrt{1-\rho^{XY}}\right)\mathbf{y}_i + \left(\sqrt{\rho^{XY}}\right) \mathbf{z}.
\end{equation*}
Thus, after the transformation, for the cross-correlated data it holds that
\begin{equation*}
  \Cor(\mathbf{x}_i, \mathbf{y}_j) = \Var\left(\left(\sqrt{\rho^{XY}}\right) \mathbf{z}\right) = \rho^{XY},
\end{equation*}
for all $i, j \in \{1,2, \ldots, s\}$.
The considered numbers of cross-correlated variables are $s = 30, 60, 90, 120, 150$. For the amount of cross-correlation we consider two choices, $\rho^{XY} = 0.5$ and $\rho^{XY} = 0.9$ (we restrict the presentation to positive cross-correlations, because negative cross-correlations did not yield any noticeable changes in the simulation results). The simulation with each combination of parameters is performed 1000 times.

We apply Procedure~\ref{proc:FDRcorrectedSCCA} with a target level $q = 0.1$ to this non-Gaussian dataset. In Step 2 of the procedure, we tune the penalty parameters such that the preliminary estimates $\widehat{\fatu}^{(0)}$ and $\widehat{\fatv}^{(0)}$ each contain approximately $\frac{n/3}{2} \approx 160$ non-zero entries (see Section~\ref{sec:gaussian_data} for a discussion of this choice). Figure~\ref{fig:hybrid_simulation_FDR_combined} shows that the estimated $\FDR\left(\widehat{\fatu}\right)$ and $\FDR\left(\widehat{\fatv}\right)$ levels are always nearly equal to or below the target upper bound $q$. The FDR values are close to $q$ when $s$ is substantially below 160, which is the enforced number of non-zero entries in $\widehat{\fatu}^{(0)}$ and $\widehat{\fatv}^{(0)}$. Figure~\ref{fig:hybrid_simulation_TPR_combined} shows the estimated $\TPR\left(\widehat{\fatu}\right)$ and $\TPR\left(\widehat{\fatv}\right)$. As one would expect, the TPR values increase with an increase in $s$ or in $\rho^{XY}$.
One drawback seen in Figure~\ref{fig:hybrid_simulation_TPR_combined} is that our method is very low-powered in extremely sparse settings, more so than the results from Section~\ref{sec:gaussian_data}, which is explained by the higher dimensionality of the data here.
To summarize, we conclude that the proposed procedure retains its FDR-controlling properties on this non-Gaussian dataset, which is structured similarly to the real data analysed in the next section. This gives us confidence on the application of the proposed procedure to real data.

\begin{figure*}[ht]
\centering
\subfloat[Estimated FDR by sparsity of the true solution.]{\includegraphics[width=0.45\textwidth]{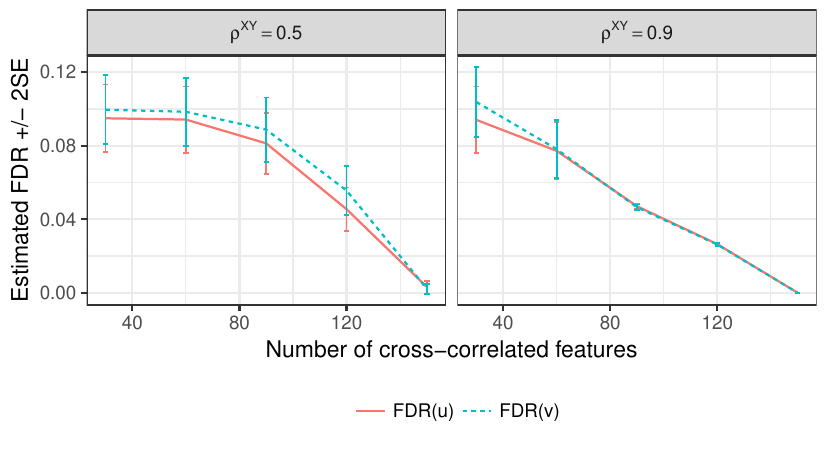}%
\label{fig:hybrid_simulation_FDR_combined}}
\hfill
\subfloat[Estimated TPR by sparsity of the true solution.]{\includegraphics[width=0.45\textwidth]{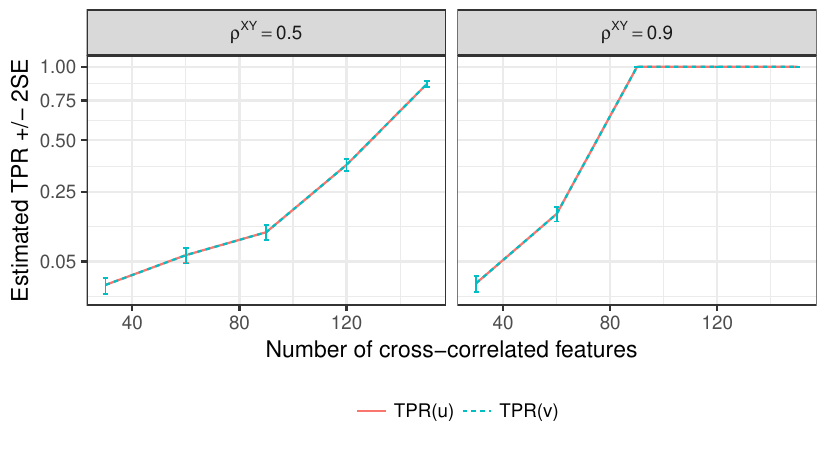}%
\label{fig:hybrid_simulation_TPR_combined}}
\caption{The proposed FDR-corrected sparse CCA, Procedure~\ref{proc:FDRcorrectedSCCA}, was applied to ``hybrid'' data, which consist of a matrix $X\in\R^{956\times 10000}$ based on real SNP data, and a matrix $Y\in\R^{956\times 10000}$ with random Gaussian entries. The parameter $\rho^{XY}$ signifies the amount of cross-correlation between features of $X$ and $Y$.
(a) The estimated $\FDR(\widehat{\fatu})$ and $\FDR(\widehat{\fatv})$ consistently stay below the target upper bound $q = 0.1$ regardless of the number of truly cross-correlated features. The FDR values decrease as the true number of cross-correlated features approaches the bound ($\frac{n/3}{2} \approx 160$) placed on the number of non-zero entries in $\widehat{\fatu}^{(0)}$ and $\widehat{\fatv}^{(0)}$.
(b) The estimated $\TPR\left(\widehat{\fatu}\right)$ and $\TPR\left(\widehat{\fatv}\right)$. TPR values increase with $\rho^{XY}$, and with an increase in the number of cross-correlated features present in the data.
}
\label{fig:hybrid_simulation}
\end{figure*}

\section{Application to imaging genomics}
\label{sec:upenn_data}

This section presents an application of Procedure~\ref{proc:FDRcorrectedSCCA}, the proposed FDR-corrected sparse CCA procedure, to a large imaging genomics dataset.
Our goal is to identify genomic regions which exhibit a significant relationship with brain functional connectivity measures, and vice versa.
In Section~\ref{sec:data_acquisition}, a brief description of the data as well as of the performed pre-processing steps is followed by a discussion of the presence of possible confounding factors in the data. Finally, in Section~\ref{sec:FDRcorrectedSCCA_on_PNC}, we apply the proposed procedure to the dataset, and discuss the obtained results.

\subsection{Data acquisition, pre-processing, and exploratory analysis}
\label{sec:data_acquisition}

The Philadelphia Neurodevelopmental Cohort \cite{Satterthwaite2014-jy} (PNC) is a large-scale collaborative study between the Brain Behaviour Laboratory at the University of Pennsylvania and the Children's Hospital of Philadelphia. It contains, among other modalities, an emotion identification (EMID) fMRI task, and SNP arrays for over 900 children, adolescents, and young adults.

The EMID task fMRI data were pre-processed using the software SPM12 \cite{Ashburner2014-xn}. The performed pre-processing steps included an adjustment for head movement, spatial normalization to standard MNI adult template (spatial resolution of $2\mm\times 2\mm\times 2\mm$), and a spatial smoothing with a $3\mm$ FWHM Gaussian kernel. The influence of motion (6 parameters) was further addressed using a regression procedure, and the functional time series were band-pass filtered using a 0.01--0.1Hz frequency range.
The pre-processed data were decomposed into subject-specific spatial maps and time courses using a group-level spatial ICA as implemented in the GIFT toolbox\footnote{\url{http://mialab.mrn.org/software/gift/}}. The number of components in group ICA was set to $C = 100$. For each subject $i \in \{1,2,\dots,n\}$ a functional connectome (FC) was estimated from the $C$ subject-specific time courses as their $C \times C$ sample covariance matrix.
Generally a subject's FC can be obtained by calculating Pearson correlations between the blood-oxygenation-level dependent (BOLD) time series corresponding to each pair of regions of interest (ROI) in the subject's brain.
Here we chose to compute the stationary FC on group ICA time courses rather than on the BOLD time series directly, in order to consolidate the relevant signals within a smaller number of components.
Each subject's estimated FC was flattened into a vector $\faty_i$ of length $\frac{C(C-1)}{2} = 4950$. The matrix $Y \in \R^{n\times \frac{C(C-1)}{2}}$ was constructed as the concatenation of the vectors $\faty_i$ as its rows, where $i \in \{1,2,\dots,n\}$.

The SNP data available in PNC were acquired using 6 different platforms. We kept only the subjects genotyped by the four most commonly used platforms (all manufactured by Illumina). The PLINK software \cite{Purcell2007-bx} was used for SNP data pre-processing, including Hardy-Weinberg equilibrium tests with significance level $1\mathrm{e}{-5}$ for genotyping errors, extraction of common SNPs (MAF $>5\%$), linkage disequilibrium pruning with threshold of $0.9$ for highly correlated SNPs, and removal of SNPs with missing call rates $>10\%$ and samples with missing SNPs $>5\%$. The remaining missing values were imputed by Minimac3 \cite{Das2016-ei} using the reference genome from 1000 Genome Project\footnote{\url{http://www.1000genomes.org}} Phase 1 data.
The resulting dataset contains values for $98804$ SNPs, representing the presence of two minor alleles by 2, one minor allele by 1, and the absence of a minor allele by 0.

Of course the proposed FDR-corrected sparse CCA procedure addresses only one source of false discoveries -- the multiple testing problem.
In practice, false discoveries may arise due to other reasons as well, because the structure of real world data is typically more complex than the idealized mathematical description, and not all mathematical assumptions will be met.
For example, in genomic studies, allele frequencies as well as the size and composition of linkage disequilibrium blocks may significantly differ between groups of subjects due to systematic ancestry differences, which may lead to spurious associations.
Since the PNC samples are multi-ethnic and were acquired with multiple genotyping platforms, this issue needs to be properly addressed here.
If the variability in ethnicity or genotyping platform (or some other unobserved external factor) affects only the SNP data, but does not affect the fMRI data, we expect our method to perform similarly well as observed empirically in the simulations of Section~\ref{sec:hybrid_simulation}, which verified the good behaviour of our method when applied to the real SNP data as $X$ and a synthetic dataset as $Y$.
However, the simulations of Section~\ref{sec:hybrid_simulation} do not cover the situation when a variable is related to both the SNP data and the fMRI data simultaneously, in which case it is in danger of becoming a confounder and yielding false discoveries.
Commonly, principal component analysis (PCA) is used within an exploratory analysis context to investigate the presence of such confounding factors in genomic studies (e.g., see \cite{Price2006-fn, Leek2010-ct}).
After performing PCA on the SNP matrix and on the FC matrix, we investigated, in both cases, for potential relationships between the first 20 principle components (PCs) and external variables such as ethnicity, genotyping platform, and age.
No external variables appeared to be strongly related to both the SNP values and the FC features simultaneously by the performed exploratory analyses.\footnote{We provide a summary of these results here; see Supplementary Materials for more detail (available in the supplementary files/multimedia tab).}
When plotting PCs against each other, it is evident that ethnicity differences are strongly related to the first three SNP-specific PCs, whereas no such relationships were observed for the fMRI data. Similarly, a PCA analysis of the considered fMRI features indicated a clear relationship between the subjects' FC and age values, but (as expected) age shows no relationship with the SNP data. No noticeable associations with the genotyping platform differences or gender were observed from PCA analysis for the SNP data or the considered fMRI features in this dataset.
Since no factors seem to have a clear relationship to both the genomic and the brain imaging data simultaneously, we proceeded with the application of the proposed FDR-corrected SCCA procedure to this dataset (we however revisit this issue once a set of features is identified by our method).

\subsection{Identification of cross-correlations between genomic and functional brain connectivity features on PNC data}
\label{sec:FDRcorrectedSCCA_on_PNC}

We apply Procedure~\ref{proc:FDRcorrectedSCCA} to a matrix derived from the PNC SNP data and a matrix derived from the FC values resulting from the PNC EMID task fMRI data.
The subset of PNC subjects for whom we have both types of data available is of size $n = 811$.
Following the process applied in the simulations of Section~\ref{sec:hybrid_simulation}, in this analysis we tune the penalty parameters in Step 2 of Procedure~\ref{proc:FDRcorrectedSCCA} such that $\widehat{\fatu}^{(0)}$ and $\widehat{\fatv}^{(0)}$ each contain approximately $\frac{n/3}{2} \approx 135$ non-zero entries.
Because we aim to capture all relevant genomic and brain imaging features within this limited number of non-zero entries of $\widehat{\fatu}^{(0)}$ and $\widehat{\fatv}^{(0)}$ respectively, we perform some feature engineering to focus the relevant signals within a smaller number of transformed features.
For the SNP data this is accomplished by gene-specific PCA transformations, as described in Section~\ref{sec:hybrid_simulation}.
This reduces the dimensionality of the untransformed SNP data to $p_X = 60372$ gene-level features, which form the matrix $X$.
In addition to consolidating the genomic variability within a smaller number of derived features, the utilized approach also facilitates the interpretation of the findings, because each column of $X$ directly corresponds to exactly one gene.
The fMRI data is transformed to functional connectivity profiles, as described in Section~\ref{sec:data_acquisition}, whereby we obtain FC measures on group ICA time courses rather than on the fMRI time series directly, in order to consolidate the relevant signals within a smaller number of components. The resulting matrix $Y$ has $p_Y = 4950$ columns (FC features).
The columns of $X$ and $Y$ are standardized to have sample means equal to zero, and sample standard deviations equal to one.
We apply Procedure~\ref{proc:FDRcorrectedSCCA} with an upper bound of $q = 0.1$ on the false discovery rates.
When we split the data into subsets in Step 1 of Procedure~\ref{proc:FDRcorrectedSCCA}, we ensure that the empirical joint distribution of variables representing ethnicity, genotyping platform, age group, and gender is consistent across the generated subsets.
The obtained FDR-corrected sparse CCA solution $\left\{\widehat{\fatu}, \widehat{\fatv}\right\}$ includes 129 genomic features and 107 FC features.

The sparse CCA results can help explain how genetic variation between subjects influences the subjects' functional brain connectivity as estimated from fMRI during an emotion identification task.
For example, consider the top 10 most significant genomic features (when ranked by p-values from Step 5 of Procedure~\ref{proc:FDRcorrectedSCCA}). As summarized in Table~\ref{tab:top10genes}, for each gene we found previous studies of association with cognitive ability, brain development or connectivity, or neurodevelopmental disorders. This was expected given that a cognitive task was performed by the subjects during the MRI acquisition and that the PNC subjects belong to an age range (8-21 years) at which crucial processes of neurodevelopment are taking place.
In fact, a recent study \cite{Kaufmann2017-kp} observes a notable variability in the PNC cohort with respect to brain dysfunction, where in addition to a group of healthy controls consisting of 153 PNC subjects, three groups were found to have increased (prodromal) symptoms of attention deficit disorder (107 PNC subjects), schizophrenia (103 PNC subjects), and depression (85 PNC subjects).

The FC features describe pair-wise interactions of group ICA components obtained from fMRI data (see Section~\ref{sec:data_acquisition}). By extracting within each spatial component only the ROI corresponding to the activation coefficient with the largest absolute value, we can map the selected FC features to pairs of ROI (we use the definitions of the AAL parcellation map \cite{Tzourio-Mazoyer2002-mw}). The top 10 most significant FC features (according to p-values from Step 5 of Procedure~\ref{proc:FDRcorrectedSCCA}) are shown in Table~\ref{tab:top10FC}.
The analysis implies that activations in these ROI correspond to the processing of emotional faces, and are significantly regulated by the genomic differences across subjects.
For the EMID task in PNC fMRI data the subjects viewed images of emotional faces and were asked to label the emotion displayed.
A wide network of ROI is engaged during this task, because human faces are complex stimuli which require the performance of many subtasks such as basic visual processing, identification of other individuals, and extraction of nonverbal and affective information.
Previous fMRI studies associate many of the detected ROI with the identification or processing of emotions.
We find support in previous studies for all of the ROI shown in Table~\ref{tab:top10FC}.
Studies of BOLD fMRI response to images of emotional faces identify increased positive activations in the superior temporal gyri and the middle occipital gyri in combination with negative activations in the superior frontal gyri \cite{Radua2010-xl}, positive activation in the middle frontal gyrus for individuals with schizophrenia \cite{Radua2010-xl}, and neural activity in the right angular gyrus in young subjects \cite{Iidaka2002-ba}.
Additionally, it has been shown that the anterior cingulate cortex operates together with the lateral and medial prefrontal cortex (including the superior frontal gyrus) in order to regulate both cognitive and emotional processing \cite{Bush2000-pp, Etkin2011-gw}. Emotional processing has also been previously associated with activations in the supplementary motor area \cite{Etkin2011-gw}.
The superior temporal gyrus is also known to be associated with autism \cite{Bigler2007-hg} (a condition that impairs the processing of emotional faces), while the calcarine fissure is involved in processing of visual information, and was also found to be selectively activated during the processing of faces in association with the intensity of the displayed emotion \cite{Morris1998-ht}.

\begin{table}[ht]
  \begin{center}
    \begin{tabular}{lp{6cm}}
      \hline
       Gene & Previously studied in association with \\
      \hline
      DAB1 & Autism spectrum disorder: \cite{Fatemi2005-gb, Li2013-zn}; schizophrenia: \cite{Impagnatiello1998-bl, Verbrugghe2012-id}; brain development: \cite{Herrick2004-dg, Pramatarova2008-xf, Feng2009-to}. \\
      NAV2 & Brain biology \& development: \cite{Marzinke2013-nq}. \\
      WWOX & Cognitive ability: \cite{McClay2011-gz}; brain development: \cite{Tabarki2015-cb}. \\
      CNTNAP2 & Autism: \cite{Alarcon2008-ke, Penagarikano2012-lq}; brain connectivity: \cite{Dennis2011-ye}; brain development: \cite{Tan2010-fp}; schizophrenia \& major depression: \cite{Ji2013-fx}; cognitive ability (linguistic processing): \cite{Whalley2011-cj, Whitehouse2011-wf}. \\
      NELL1 & Brain biology \& development: \cite{Kuroda1999-zk}. \\
      PTPRT & Brain biology \& development: \cite{Ensslen-Craig2004-ta, McAndrew1998-ez, Besco2006-id}. \\
      FHIT & Cognitive ability: \cite{Davis2010-lz}; autism spectrum disorders: \cite{Corominas2014-cu, Salyakina2010-aw}; attention deficit hyperactivity disorder: \cite{Lasky-Su2008-ou}. \\
      MACROD2 & Autism spectrum disorder: \cite{Anney2010-ds, Curran2011-ho, Jones2014-rp}. \\
      LRP1B & Cognitive function: \cite{Poduslo2010-we}. \\
      DGKB & Brain biology \& development: \cite{Hozumi2009-uf}; bipolar disorder: \cite{Caricasole2002-yw, Kakefuda2010-mh}. \\
      \hline
    \end{tabular}
  \end{center}
  \caption{Top 10 most significant among the identified genes.}
  \label{tab:top10genes}
\end{table}


\begin{table}[ht]
  \begin{center}
    \begin{tabular}{p{4cm}p{4cm}}
      \hline
      ROI \#1 & ROI \#2 \\
      \hline
      anterior cingulate and paracingulate gyri (left) & superior frontal gyrus, medial (right) \\
      superior temporal gyrus (left) & superior temporal gyrus (left)\\
      superior frontal gyrus, medial (right) & superior frontal gyrus, medial (left)\\
      middle occipital gyrus (left) & superior frontal gyrus, medial (left) \\
      supplementary motor area (left) & anterior cingulate and paracingulate gyri (left)\\
      supplementary motor area (left) & middle occipital gyrus (left) \\
      middle temporal gyrus (left) & middle occipital gyrus (left) \\
      middle frontal gyrus, orbital part (left) & superior frontal gyrus, medial (left) \\
      middle occipital gyrus (right) & superior frontal gyrus, medial (right) \\
      calcarine fissure and surrounding cortex (right) & angular gyrus (right) \\
      \hline
    \end{tabular}
  \end{center}
  \caption{Top 10 most significant FC features mapped to ROI pairs.}
  \label{tab:top10FC}
\end{table}

Similarly we were able to find previous work that supports our findings for many other identified genes and brain regions. But because the rather large number of identified features prohibits the presentation of a full literature research for all of them, and because in this work we mainly focus on methodology development rather than on biological insights, we omit a further coverage of the possible biological implications of the CCA results here.
However, we briefly revisit the concern expressed in Section~\ref{sec:data_acquisition} regarding potential confounding effects in the application of our sparse CCA method to the PNC data.
To investigate whether ethnicity differences may have substantially influenced the results, we extract a subsample consisting of all African American and Caucasian/white subjects from the analyzed data, consisting of 313 and 406 subjects respectively, while any other distinct ethnicity group had less than $10$ subjects, and therefore was excluded from the following investigation.
On this subsample we obtain an individual p-value of association with ethnicity for each genomic and for each FC feature using the ANOVA framework.
A subset containing only the ANOVA p-values corresponding to the selected genomic (or FC) features is formed. Then a Kolmogorov-Smirnov test with significance level 0.05 is used to compare the empirical distribution of p-values within this subset to the empirical distribution of the whole set of ANOVA p-values, corresponding to the whole set of genomic (or respectively FC) features.
This comparison yields that the genomic as well as the FC features selected as significant by our method indeed tend to exhibit a much stronger association with ethnicity than other genomic or FC features.
Thus there is a possibility that some of the identified genomic and FC features may have been determined by systematic ancestry differences in the sample.
The fact that we can find previous work that supports our finding provides evidence to the claim of no ethnic confounding, but unfortunately this question cannot be fully resolved without an analysis on new data.
Using an analogous approach with respect to other external factors such as genotyping platform, age, and gender, no similar effects have been observed, supporting the initial conclusion of Section~\ref{sec:data_acquisition} that these factors are not confounding the analysis.

To conclude, we have identified sparse sets of genes and FC features carrying a substantial amount of cross-correlation in a multi-ethnic cohort\footnote{Also of note is that the ethnicity distribution in PNC is similar to much of the North American population.}, while correcting for the type of false discoveries which appear as cross-correlated due to random noise in a high-dimensional problem.
These results, however, should be taken with care, because some of the identified cross-correlated features may represent indirect relationships by proxy of systematic ancestry differences.

\section{Discussion}
\label{sec:discussion}

As discussed in Section~\ref{sec:introduction}, model selection in sparse CCA, especially as it pertains to the unknown sparsity level, is largely an unsolved problem.
In this work we have proposed a notion of FDR in the context of sparse CCA as a criterion for the selection of an appropriate sparsity level. At the same time we have proposed an FDR-corrected sparse CCA method, which is adaptive to the unknown true sparsity of the canonical vectors. With theoretical arguments and a series of simulation studies we have shown that the proposed method controls the FDR of sparse CCA estimators at a user-specified level.

Interestingly, even though our method was analytically derived under Gaussian assumptions, in the considered non-Gaussian simulations it still achieved FDR control. Thus it would be interesting to extend the theoretical argumentation for FDR control of the proposed method beyond the Gaussian scenarios.
It would also be valuable to investigate better ways of selecting the sparsity constraints on the preliminary estimates $\widehat{\fatu}^{(0)}$ and $\widehat{\fatv}^{(0)}$, such that all cross-correlated features of $X$ and $Y$ are identified.
Especially in the high-dimensional case a trade-off arises, because when $\widehat{\fatu}^{(0)}$ and $\widehat{\fatv}^{(0)}$ are too dense, the applicability of the asymptotic distribution from Theorem~\ref{thm:XtYv_is_asymptotically_normal} becomes questionable, and the FDR-correction step can significantly lose power.
In fact, the FDR-correction step in our procedure would improve substantially, if we had a theoretical result which is analogous to Theorem~\ref{thm:XtYv_is_asymptotically_normal} but regards $\fatv$ as a random variable and considers the asymptotic behaviour of $X^T Y \fatv$ as $p_X, p_Y \to \infty$.
Furthermore, especially in high-dimensional scenarios, the proposed method tends to be too conservative, having an FDR substantially below the nominal level. An adjustment that would result in an FDR equal to the nominal level is very desirable, because it would increase the detection power.
These are the important theoretical questions that we leave for future research.
As to the application in imaging genomics, there is clearly a potential for improvement in how the genomic and brain fMRI data are represented as features within $X$ and $Y$. In Section~\ref{sec:FDRcorrectedSCCA_on_PNC} we transformed the SNP data using PCA within each gene, and applied group ICA to the fMRI data, before the application of sparse CCA\@. This allowed us to select features at the gene and ROI level.
Alternatively, one could apply sparse CCA to the untransformed $\{0,1,2\}$-valued SNP data and to voxel-wise fMRI contrasts, and then during the FDR-correction step one could consolidate the voxel specific p-values into a single p-value for each ROI, and the per-SNP p-values into gene-based p-values, in order to obtain gene and ROI level results with possibly a higher detection power. Such an approach can be based on a statistically rigorous foundation by including a method analogous to VEGAS \cite{Liu2010-tb} (or a similar gene-based technique) as an additional step in Procedure~\ref{proc:FDRcorrectedSCCA}.

In conclusion, we hope that this work can motivate a scholarly conversation and further research about FDR control in the context of sparse CCA.


%

\appendices
\section{Asymptotic Normality proof}
\label{sec:XtYv_is_asymptotically_normal_proof}

In this Section we present the proof of Theorem~\ref{thm:XtYv_is_asymptotically_normal}. We need the following Lemma.

\begin{lemma}
  Assume that $A\in\R^{p\times p}$ and $B\in\R^{p\times p}$ are symmetric matrices. Let $S\in\R^{p\times p}$ be a random matrix that follows a Wishart distribution with parameters $\Sigma\in\R^{p\times p}$, $p$ and $n$. Then it holds that
  \begin{equation*}
    \Cov\left( \tr(AS), \tr(BS) \right) = \tr\left( 2n A \Sigma B \Sigma \right).
  \end{equation*}
  \label{lemma:covariance_of_traces}
\end{lemma}
\begin{proof}
  Because $A$ and $B$ are symmetric, we have that
  \begin{small}
  \begin{align*}
    \Cov\left( \tr(AS), \tr(BS) \right) &=
    \Cov\left( \sum_{i,j = 1,\dots,p} a_{ij} s_{ij}, \sum_{k,l = 1,\dots,p} b_{kl} s_{kl} \right)\\
    &= \sum_{i,j = 1,\dots,p} a_{ij} \sum_{k,l = 1,\dots,p} b_{kl} \Cov\left( s_{ij}, s_{kl} \right)\\
    &= \sum_{i,j = 1,\dots,p} a_{ij} \left[ \mathrm{mat} \left( \Cov(\vectorization(S)) \cdot \vectorization(B) \right) \right]_{ij}\\
    &= \tr\left( A \cdot \mathrm{mat}\left( \Cov(\vectorization(S)) \vectorization(B) \right) \right),
  \end{align*}
  \end{small}
  where $\vectorization(\cdot)$ denotes vectorization, and $\mathrm{mat}(\cdot)$ denotes matricization (i.e., the inverse of vectorization).

  Now, Proposition 8.3 in \cite{Eaton2007} together with Kronecker product properties implies that
  \begin{align*}
    \Cov\left( \tr(AS), \tr(BS) \right) &= \tr\left( A \cdot \mathrm{mat}\left( (2n \Sigma \otimes \Sigma)  \vectorization(B) \right) \right)\\
    &= 2n \tr\left( A \Sigma B \Sigma \right).
  \end{align*}
\end{proof}

\subsection{Proof of Theorem~\ref{thm:XtYv_is_asymptotically_normal}}

\begin{proof}
  Because the random vectors
  $\left(
    \fatx^{(1)},
    \faty^{(1)}
  \right)$,
  $\left(
    \fatx^{(2)},
    \faty^{(2)}
  \right)$, $\dots$,
  $\left(
    \fatx^{(n)},
    \faty^{(n)}
  \right)$ are independent and identically distributed, and since
  \begin{equation*}
    \frac{1}{n} X^T Y \fatv
    = \frac{1}{n} \sum_{k=1}^n \fatx^{(k)} \cdot \left(\faty^{(k)}\right)^T \cdot \fatv,
  \end{equation*}
  the multidimensional version of the Central Limit Theorem immediately gives
  \begin{equation*}
    \sqrt{n} \left(\frac{1}{n} X^T Y \fatv - \fatmu \right) \overset{\mathcal{D}}{\longrightarrow} \mathcal{N}(0, \Omega),
  \end{equation*}
  where
  \begin{align}
    \fatmu &= \E\left( \fatx^{(k)} \cdot \left(\faty^{(k)}\right)^T \cdot \fatv \right),
    \label{eq:fatmu} \\
    \Omega &= \Cov\left( \fatx^{(k)} \cdot \left(\faty^{(k)}\right)^T \cdot \fatv \right),
    \label{eq:Omega}
  \end{align}
  for any $k \in\left\{ 1,2,\dots,n \right\}$.

  Now, it is left to prove that the entries of $\fatmu$ and $\Omega$ have the form specified in equations (\ref{eq:XtYv_mean}) and (\ref{eq:XtYv_covariance}).

  Equation (\ref{eq:XtYv_mean}) follows directly from the linearity of expectation and the fact that all entries of $X$ and $Y$ are Gaussian with mean zero.

For notational convenience, denote $A := \begin{bmatrix} X & Y \end{bmatrix}$ and $\fatw := \frac{1}{n} X^T Y \fatv$. In order to prove (\ref{eq:XtYv_covariance}), we first observe that
\begin{equation*}
  A^T A \sim W(\Sigma, p_X + p_Y, n),
\end{equation*}
which denotes a Wishart distribution with parameters $\Sigma$, $(p_X + p_Y)$ and $n$, where $\Sigma$ is defined by equation (\ref{eq:definition_of_Sigma}). We also observe that the entries of $\fatw$ can be written as
  \begin{equation*}
     w_i = \tr\left( \frac{1}{2n} \begin{bmatrix} 0 & \fate_i \fatv^T \\ \fatv \fate_i^T & 0 \end{bmatrix} \cdot A^T A \right),
  \end{equation*}
  where $\fate_i \in \R^p$ denotes the $i$th standard basis vector in $\R^{p_X}$.

Thus, using Lemma~\ref{lemma:covariance_of_traces}, we conclude that
\begin{small}
\begin{align}
  &\Cov(w_i, w_j) = \nonumber \\
  &= 2n \cdot \tr\left(
  \frac{1}{2n} \begin{bmatrix}
    0 & \fate_i \fatv^T \\
    \fatv \fate_i^T & 0 \end{bmatrix}
  \cdot
  \Sigma
  \cdot
  \frac{1}{2n} \begin{bmatrix}
    0 & \fate_j \fatv^T \\
    \fatv \fate_j^T & 0 \end{bmatrix}
  \cdot
  \Sigma
  \right) \nonumber \\
  &= \frac{1}{2n} \tr\left(
  \begin{bmatrix}
    \fate_i \fatv^T \Sigma_{XY}^T & \fate_i \fatv^T \Sigma_Y \\
    \fatv \fate_i^T \Sigma_X & \fatv \fate_i^T \Sigma_{XY}
  \end{bmatrix}
  \cdot
  \begin{bmatrix}
    \fate_j \fatv^T \Sigma_{XY}^T & \fate_j \fatv^T \Sigma_Y \\
    \fatv \fate_j^T \Sigma_X & \fatv \fate_j^T \Sigma_{XY}
  \end{bmatrix}
  \right) \nonumber \\
  &= \frac{1}{2n} \left\{
    \tr\left( \fate_i \fatv^T \Sigma_{XY}^T \fate_j \fatv^T \Sigma_{XY}^T \right)
    +
    \tr\left( \fate_i \fatv^T \Sigma_Y \fatv \fate_j^T \Sigma_X \right)
    \right. \nonumber \\
    &\quad\quad \left. +
    \tr\left( \fatv \fate_i^T \Sigma_X \fate_j \fatv^T \Sigma_Y \right)
    +
    \tr\left( \fatv \fate_i^T \Sigma_{XY} \fatv \fate_j^T \Sigma_{XY} \right)
  \right\} \nonumber \\
  &= \frac{1}{2n} \left\{
    2 \tr\left( \fate_i \fatv^T \Sigma_{XY}^T \fate_j \fatv^T \Sigma_{XY}^T \right)
    +
    2 \tr\left( \fatv^T \Sigma_Y \fatv \fate_j^T \Sigma_X \fate_i \right)
  \right\} \nonumber \\
  &= \frac{1}{n} \left\{
  \left( \sum_{k=1}^{p_Y} v_k \rho^{XY}_{j,k} \right) \left( \sum_{k=1}^{p_Y} v_k \rho^{XY}_{i,k} \right)
    +
    \fatv^T \Sigma_Y \fatv \rho^X_{j,i}
  \right\}
  \label{eq:proof_of_covariance_of_XtYv}
\end{align}
\end{small}
Moreover, we have that
\begin{align}
  \Cov(\fatw)
  &= \Cov\left( \frac{1}{n} \sum_{k=1}^n \fatx^{(k)} \cdot \left(\faty^{(k)}\right)^T \cdot \fatv \right) \nonumber \\
  &= \frac{1}{n^2} \sum_{k=1}^n \Cov\left( \fatx^{(k)} \cdot \left(\faty^{(k)}\right)^T \cdot \fatv \right)
  = \frac{1}{n} \Omega,
  \label{eq:cov_fatw}
\end{align}
where we used equation (\ref{eq:Omega}) and the fact that the random vectors
$\left(
  \fatx^{(1)},
  \faty^{(1)}
\right)$,
$\left(
  \fatx^{(2)},
  \faty^{(2)}
\right)$, $\dots$,
$\left(
  \fatx^{(n)},
  \faty^{(n)}
\right)$ are independent and identically distributed.

Finally, equation (\ref{eq:proof_of_covariance_of_XtYv}) combined with equation (\ref{eq:cov_fatw}) yield equation (\ref{eq:XtYv_covariance}).
\end{proof}

\section*{Acknowledgment}

The work was partially supported by NIH (R01 GM109068, R01 MH104680, R01 MH107354, P20 GM103472, R01 REB020407, R01 EB006841) and NSF (\#1539067).

\ifCLASSOPTIONcaptionsoff
  \newpage
\fi



\bibliographystyle{IEEEtran}
\bibliography{IEEEabrv,FDRcorrectedCCA}
%
%
%

\end{document}